\documentclass[conference]{IEEEtran}
\IEEEoverridecommandlockouts
\usepackage{url}

\usepackage{booktabs}
\usepackage{multirow}
\usepackage{tabularx}
\usepackage{array}
\usepackage{placeins}
\usepackage{float}

\newcommand{\LMP}{\mathsf{LMP}}

\newtheorem{theorem}{Theorem}

\newcommand{\LMPU}{\mbox{LMP+U}}

\newcommand{\LMPUm}{\mathsf{LMP+U}}
\newcommand{\MDCP}{\mathsf{MDCP}}

\usepackage[export]{adjustbox}
\usepackage[caption=false]{subfig}
\usepackage{multirow}
\usepackage{cite}
\usepackage{amsmath,amssymb,amsfonts,bm,bbm}
\usepackage{algorithmic}
\usepackage{graphicx}
\usepackage{wrapfig}
\usepackage{textcomp}
\usepackage{mathtools}

\usepackage{xcolor}
\def\BibTeX{{\rm B\kern-.05em{\sc i\kern-.025em b}\kern-.08em
    T\kern-.1667em\lower.7ex\hbox{E}\kern-.125emX}}

\usepackage{etoolbox}

\newcommand{\compactdisplay}{%
  \setlength{\abovedisplayskip}{3pt plus 1pt minus 1pt}%
  \setlength{\belowdisplayskip}{3pt plus 1pt minus 1pt}%
  \setlength{\abovedisplayshortskip}{2pt plus 1pt minus 1pt}%
  \setlength{\belowdisplayshortskip}{2pt plus 1pt minus 1pt}%
  \setlength{\jot}{1pt}%
}

\AtBeginEnvironment{align}{\compactdisplay}
\AtBeginEnvironment{align*}{\compactdisplay}
\AtBeginEnvironment{equation}{\compactdisplay}
\AtBeginEnvironment{equation*}{\compactdisplay}
\AtBeginEnvironment{gather}{\compactdisplay}
\AtBeginEnvironment{gather*}{\compactdisplay}
\AtBeginEnvironment{multline}{\compactdisplay}
\AtBeginEnvironment{multline*}{\compactdisplay}

\usepackage{setspace}
\begin{document}

\title{Optimal Uniform Pricing for Multi-Interval Dispatch without Make-Whole Uplifts
\thanks{\scriptsize
Valentina Norambuena and Lang Tong (\{van27,lt35\}@cornell.edu) are with the School of Electrical and Computer Engineering, Cornell University, Ithaca, NY, USA.
Cong Chen (cong.chen@dartmouth.edu) is with the Thayer School of Engineering, Dartmouth College, Hanover, NH, USA. Timothy D. Mount (tdm2@cornell.edu) is with the Dyson School of Applied Economics and Management, Cornell University, Ithaca, NY, USA. This work is supported in part by the National Science Foundation under Awards 2603293 and 2412776, and the Power Systems Engineering Research Center (PSERC) Research Project M-48.}}

\author{
Valentina Norambuena-Guzman,~\IEEEmembership{Student Member,~IEEE},
\qquad
Cong Chen,~\IEEEmembership{Member,~IEEE},
\\
Lang Tong,~\IEEEmembership{Fellow,~IEEE},
\qquad
and Timothy D. Mount,~\IEEEmembership{Life Member,~IEEE} 
}

\maketitle

\begin{abstract}
In a network with ramp-limited generators and inaccurate net-demand forecasts, practical rolling-window dispatch can drive locational marginal prices (LMPs) below generators’ bid-in offers. In such cases, out-of-market (OOM) settlements are used to compensate generators and maintain dispatch-following incentives, but OOM can have negative consequences, including nontransparent real-time price signals, discriminatory compensation, and incentives for untruthful bidding. This paper presents an optimal uniform pricing rule that minimizes demand payments, eliminates OOM make-whole payments, preserves LMP-based congestion charges, and ensures revenue adequacy. We derive the proposed pricing rule in closed form and relate it to existing pricing schemes. Numerical comparisons demonstrate favorable generator profits and reduced price volatility. However, higher generator profits are accompanied by increased demand payments, reflecting the in-market, uniform allocation of ramping costs while preserving the LMP-based congestion charges widely used in real-time market settlements. The numerical results also show that, under LMP with OOM settlement, a price-taking generator has an incentive to inflate its offer, whereas this incentive is absent under the proposed pricing rule within the tested bid range.
\end{abstract}


\section{Introduction}
Out-of-market (OOM) settlement in an organized electricity market, also known as OOM uplift, is an intervention mechanism through which a system operator compensates market participants when market-clearing prices alone do not provide sufficient incentives to follow the operator's dispatch instructions.

 We focus on real-time economic dispatch in an energy market with ramp-constrained generators and inaccurate net-demand forecasts. To incentivize ramping support, system operators have introduced \emph{flexible ramp products} \cite{CAISO2023FRP,MISO2025FRP} to procure and price ramp resources.

Because the economic dispatch model does not fully account for ramping costs and net-demand uncertainty in price formation, market-clearing LMPs may fall below accepted generation offers, violating the uniform-price-auction principle that accepted supply should be paid no less than its offer. OOM settlements are therefore needed to cover the resulting revenue shortfalls and maintain dispatch-following incentives. In 2022, real-time OOM bid-cost-recovery payments within CAISO
totaled about \$141 million, approximately 21\% of the estimated \$660 million real-time energy cost \cite{CAISORpt:22}.

OOM settlement has negative consequences \cite{Hogan:14EJ}. Here, we highlight two relevant aspects. While LMP with OOM settlement secures dispatch-following incentives, it creates new incentive problems.  In particular, under the ideal dispatch model, a price-taking generator has no incentive to offer above its marginal cost. With OOM uplifts, a generator's offer can affect its uplift payment even when the offer does not affect the dispatch shadow prices or the resulting LMP. Consequently, a price-taking generator may have an incentive to deviate from truthful bidding to increase its OOM uplift revenue.

OOM uplift also raises issues of discriminatory settlement and market transparency. OOM payments apply only to some participants, each compensated at different levels. While LMP is uniform and transparent to all at the same location, OOM uplifts result in non-uniform \$/MWh for generators at the same location. The market also lacks a transparent price signal for ramping support.

Hogan argues in \cite{Hogan:14EJ} that ``the first-best way to address the deficiencies of the energy market is by addressing real-time pricing. Only after exhausting the consistent improvements in real-time pricing should we resort to the second-best OOM interventions.'' While lost opportunity cost (LOC) uplift cannot in general be eliminated under uniform pricing \cite{Guo21TPS}, the widely adopted make-whole uplift can, in principle, be eliminated by setting a sufficiently high uniform price at each location. However, arbitrarily raising prices to eliminate uplift can substantially increase demand payments and distort real-time price signals.

{ This paper presents a uniform pricing rule that minimizes demand payment, eliminates OOM make-whole uplift, preserves LMP-based congestion charges and is revenue adequate under the conditions
stated in Theorem~\ref{thm1}.}

\subsection{Related Work}
One of the earliest publications addressing inadequate compensation under LMP for fast-ramp resources and the resulting need for OOM uplift is \cite{ISONE:12}, followed by extensive discussions in \cite{Schiro2017procurement, Gribik2007, Zhang2009, AlAbdullah2015, Guo21TPS, Chen21TPS, Cavicchi2018}. Many uniform pricing mechanisms have been proposed to reduce, but not eliminate, OOM uplift \cite{Ela2016, Hua2019, Zhao2020, Schiro2017flexibility, Chen25PESGM, Chen26arxiv}, among which only Maximum Dispatch Cost Pricing (MDCP) \cite{Chen25PESGM, Chen26arxiv} completely eliminates OOM MWPs in uncongested networks. Defined by the highest marginal cost among dispatched generators, MDCP supports truthful bidding under the conditions established in \cite{Chen26arxiv}. However, MDCP does not naturally generalize to congested networks. Applying MDCP at generator buses and LMP at load-only buses yields extended MDCP (E-MDCP), which does not generally preserve LMP-based congestion charges and may not be revenue
adequate. These limitations motivate the network-compatible uniform pricing rule proposed in this paper.

\subsection{Summary of Contributions}
 We present the optimal uniform pricing rule for rolling-window dispatch with generator ramping constraints and arbitrary net-demand forecasting errors. By uniform pricing, we mean that all resources at the same bus face the same energy price. Optimality refers to minimizing total demand payment subject to zero OOM make-whole uplift while preserving LMP-based congestion charges and merchandising surplus.

The main contribution of this work is fourfold. First, we formulate a demand-payment minimization problem in which the candidate nodal prices are the decision variables, subject to constraints that enforce (i) zero OOM make-whole payment and (ii) preservation of LMP-based congestion charges. The latter constraint is particularly important because it prevents the price adjustment for intertemporal ramping constraints from distorting the spatial price differences already reflected in LMP, thereby preserving the associated congestion rents and merchandising surplus.

Second, we derive a closed-form solution to the demand-payment minimization problem and show that the optimal price takes the form of LMP plus a  {\em uniform price adder across locations}. We refer to the proposed pricing rule as LMP with Uniform Price Adder, abbreviated as LMP+U.

We further establish LMP and MDCP as special cases of LMP+U under the conditions of Theorem~\ref{thm1}: LMP+U reduces to LMP when the ramping constraints and lower generation bounds are nonbinding and at least one active generator is below its upper bound, and to MDCP when the transmission constraints are nonbinding. Unlike price-preserving multi-interval pricing (PMP) \cite{Schiro2017flexibility} and constraint-preserving multi-interval pricing (CMP) \cite{Hua2019}, LMP+U admits a closed-form solution and therefore requires no separate pricing optimization.

Third, simulations show that pricing schemes combined with OOM make-whole payments can incentivize a price-taking generator to inflate its offer. In particular, under LMP, PMP, and CMP with OOM make-whole payments, a price-taking generator can increase its net revenue by inflating its offer above marginal cost. Under LMP+U, a price-taking generator has no such incentive because it cannot obtain such a gain within the tested bid range.  This result applies within the tested offer range, rather than to arbitrary deviations from truthful bidding.


Finally, numerical simulations quantify the performance of LMP+U relative to LMP/PMP/CMP with OOM make-whole uplift and provide insights into the effects of network congestion. In most simulated cases, LMP+U yields higher generator net revenue and higher demand payment than the comparison benchmarks. This ordering, however, is not universal: the optimal uniform price adder can be negative, in which case LMP+U may reduce demand payment relative to LMP. The simulations further characterize how network congestion affects the magnitude of the uniform price adder and its effects on generator revenue and demand payment.

\begin{table}[!t]
\caption{Major Symbols and Notation}
\label{tab:major_symbols}
\centering
\small
\renewcommand{\arraystretch}{1.05}
\begin{tabularx}{\columnwidth}
{@{}l>{\raggedright\arraybackslash}X@{}}
\hline
\textbf{Symbol} & \textbf{Description}\\
\hline
$\mathcal H_t, W$ & Look-ahead window at interval $t$ and window length.\\
$[M],\mathcal N_m,\mathcal A_t$ & Sets of buses, generators at bus $m$, and dispatched generators at interval $t$.\\
$g_{it'}^m,\hat d_{t'}^m,D_t$ & Generator dispatch, bus net demand, and total binding-interval net demand.\\
$c_{it'}^m$ & Bid-in marginal cost of generator $i$ at bus $m$.\\
$\bm S,\bm f$ & Shift-factor matrix and line-flow-limit vector.\\
$\bm\pi_t^{\LMP},\bm\pi_t^{\LMPUm}$ & Binding-interval LMP and LMP+U price vectors.\\
$\Delta\pi_t^*$ & Optimal uniform price adder to LMP.\\
$\mbox{\sf MWP}_t,\mbox{\sf MS}_t$ & Make-whole payment and merchandising surplus.\\
\hline
\end{tabularx}
\end{table}

\section{Dispatch Model and Bid-Cost Recovery}
\subsection{Rolling-window Dispatch in Systems with Flow Limits}

Let the dispatch session be indexed by $\mathcal{H}:=\{1,\ldots,T\}$.  In interval $t$, the operator solves a $W$-interval rolling-window dispatch over the look-ahead window $\mathcal{H}_t:=\{t,\cdots,t+W-1\}$. Let $[M]:=\{1,\ldots,M\}$ denote the set of buses and $\mathcal N_m$ the set of generators at bus $m$.

Within the rolling-window horizon $\mathcal{H}_t$ at time $t$,  let $\hat{d}_{t'}^m$ be
forecast demand and $g_{it'}^m$ the dispatch of generator $i$ at bus $m \in [M]$ in interval $t'\in \mathcal{H}_t$. Define the aggregate generation at bus $m$ and the corresponding nodal generation and net-demand vectors as
\[
q_{t'}^m :=\sum_{\mathclap{i\in\mathcal N_m}}g_{it'}^m,\quad \bm q_{t'}:=(q_{t'}^m)_{m\in[M]},\quad \hat{\bm d}_{t'}:=(\hat d_{t'}^m)_{m\in[M]}.
\]
Let $c_{it'}^m$ be the generation offer (bid-in marginal cost), $(\underline r_i^m,\overline r_i^m)$ the down- and up-ramp limits, and $(\underline g_i^m,\overline g_i^m)$ the generation limits. 

Given the shift-factor matrix $\bm S$ and the line-flow limit vector $\bm f$,  the multi-interval dispatch optimization at interval $t$ is:

{\small
\vspace{-0.3cm}
\begin{align}
\min_{\{g_{it'}^m\}}\quad &\sum_{t'=t}^{t+W-1}
  \sum_{m=1}^{M}\sum_{i\in\mathcal N_m}
  c_{it'}^m g_{it'}^m \label{eq:obj}\\
\llap{\text{s.t.}\qquad} \lambda_{t'}:\quad &\sum_{m=1}^{M}\sum_{i\in\mathcal N_m}g_{it'}^m =\sum_{m=1}^{M}\hat d_{t'}^m \label{eq:balance}\\
\bm\phi_{t'}:\quad
&\bm S\left(\bm q_{t'}-\hat{\bm d}_{t'}\right) \leq\bm f \label{eq:flow}\\
\underline{\tau}_{it'}^m,\overline{\tau}_{it'}^m:\quad
&-\underline r_i^m \leq g_{it'}^m-g_{i(t'-1)}^m \leq\overline r_i^m \label{eq:inter_ramp}\\
\underline{\alpha}_{it'}^m,\overline{\alpha}_{it'}^m:\quad
&\underline g_i^m \leq g_{it'}^m \leq\overline g_i^m. \label{eq:capacity}
\end{align}}

\noindent  Constraints~\eqref{eq:balance}--\eqref{eq:flow} apply for every $t'\in\mathcal H_t$, whereas \eqref{eq:inter_ramp}--\eqref{eq:capacity} apply for every $m\in[M]$, $i\in\mathcal N_m$, and $t'\in\mathcal H_t$.
 The solution of the above optimization gives the realized dispatch $(g^{m*}_{it})$, power balance and congestion shadow prices, $\lambda^*_t$  and $\bm \phi^*_t$ respectively, in the binding interval $t$. The $M$-dimensional {\em rolling-window LMP} vector  $\boldsymbol{\pi}_t^{\mbox{\tiny\sf LMP}}$ at all buses  is given by
\begin{equation}
\boldsymbol{\pi}_t^{\mbox{\tiny\sf LMP}} = \lambda^*_t\mathbf{1} - \bm S^{\mathsf T} \bm \phi^*_t,
\label{eq:lmp_decomposition}
\end{equation}
where $\lambda_t^*$ is the system energy component and
$-\bm S^{\mathsf T}\bm\phi_t^*$ is the congestion component.



\subsection{Bid-cost Recovery and Make-Whole Payments}
\label{subsec:underpayment}
We assume linear energy offers and fixed commitment; recovery of startup and no-load costs is outside our scope. Rolling-window LMP may fail to provide price support because $\pi^{\LMP}_{mt}$ can be lower than the bid-in cost $c^m_{it}$ of a generator $i$ at bus $m$. In such cases, the generator is typically compensated through an OOM make-whole payment.
\begin{align}
    \mbox{\sf MWP}_{it}^{m} := \max\left\{ \left(c_{it}^{m}-\pi_{mt}^{\LMP}\right) g_{it}^{m*},0 \right\}.
\label{eq:interval_mwp}
\end{align}
This out-of-market payment guarantees interval-level {energy} bid-cost recovery.\footnote{In practice, the operator may instead use session-level bid-cost recovery, which is no greater than the corresponding interval-level compensation.} The settlement is discriminatory because MWPs are paid only to generators with bid-cost shortfalls.



\section{Optimal Uniform Pricing without Make-Whole Uplift}\label{sec:lmpwupu}
This section presents a pricing optimization that minimizes the
binding-interval demand payment while eliminating OOM make-whole
payments and preserving the LMP congestion charges. The constraints imply that every feasible price vector differs from LMP by a common, possibly negative, scalar. This uniform price adder also preserves the LMP merchandising surplus, as established in Theorem~\ref{thm1}.


Consider the rolling-window dispatch at interval $t$ with binding net-demand vector $\hat{\bm d}_t$, dispatch $\bm g_t^*$, and LMP vector $\bm\pi_t^{\LMP}$. Define the dispatched-generator set and
total binding-interval net demand\footnote{{Energy balance and nonnegative generation exclude $D_t<0$. At $D_t=0$, no generation or MWP is needed, and LMP can be retained.}} as
\begin{align}
\mathcal A_t &:= \left\{(i,m):m\in[M],\ i\in\mathcal N_m,\; g_{it}^{m*}>0\right\}, \nonumber\\
D_t &:= \sum_{m\in[M]}\hat d_t^m.
\label{eq:active_set_demand}
\end{align}

For a nodal price vector $\bm p_t$, define $\mbox{\sf MS}_t(\bm p_t):= \sum_{m\in[M]}p_{mt}(\hat d_t^m-q_t^{m*})$.
In particular, $\mbox{\sf MS}_t(\bm\pi_t^{\LMP})\geq0$ \cite{Chen21TPS}.

The price optimization program is

{\small
\vspace{-0.3cm}
\begin{align}
\min_{\bm\pi_t\in\mathbb R^M}\quad &\sum_{m\in[M]}\pi_{mt}\hat d_t^m \label{eq:nodal_upmu_obj}\\
\text{s.t.}\quad &\pi_{mt}\geq c_{it}^m, &&\forall(i,m)\in\mathcal A_t, \label{eq:nodal_zero_uplift}\\
&\pi_{mt}-\pi_{nt} = \pi_{mt}^{\LMP}-\pi_{nt}^{\LMP}, &&\forall m,n\in[M]. \label{eq:congestion_spread}
\end{align}
}
Constraint~\eqref{eq:nodal_zero_uplift} enforces zero OOM make-whole payment, whereas \eqref{eq:congestion_spread} preserves the LMP congestion charges. The resulting optimal price vector, referred to as LMP+U, is characterized next.

\begin{theorem}[Optimality and Properties]
\label{thm1}
Suppose $\mathcal A_t\neq\emptyset$ and $D_t>0$. The unique optimal solution of \eqref{eq:nodal_upmu_obj}--\eqref{eq:congestion_spread} is obtained from a scalar price adder $\Delta\pi_t^*$ that is uniform across all locations:

\vspace{-0.3cm}
{\small
\begin{equation}
\bm\pi_t^{\LMPUm} = \bm\pi_t^{\LMP}+\Delta\pi_t^*\boldsymbol 1,
\ \Delta\pi_t^* := \max_{(i,m)\in\mathcal A_t}
\left(c_{it}^m-\pi_{mt}^{\LMP}\right). \label{eq:LMPU} \end{equation}
}
LMP+U has the following properties:
\begin{enumerate}
\item Every dispatched generator recovers its {submitted energy cost}, and the congestion charges under LMP+U equal those under LMP.

\item {LMP+U preserves the nonnegative LMP merchandising surplus and is therefore revenue adequate for the modeled real-time energy settlement.}\footnote{{For the energy settlement considered here, a pricing policy $\Pi$ is revenue adequate if demand payments cover generator energy payments and MWPs, i.e.,
$\mbox{\sf MS}_t(\bm\pi_t^\Pi)-\operatorname{MWP}_t^\Pi\geq0$.}}

\item If the transmission constraints in \eqref{eq:flow} are nonbinding, LMP+U reduces to MDCP. If the ramping constraints in \eqref{eq:inter_ramp} and the lower generation bounds in \eqref{eq:capacity} are nonbinding, and some $(i,m)\in\mathcal A_t$ satisfies $g_{it}^{m*}<\overline g_i^m$, then LMP+U reduces to LMP.
\end{enumerate}
\end{theorem}

\noindent{\it Proof:}
By \eqref{eq:congestion_spread}, the deviation from LMP is equal at every bus. Hence, for some $\Delta_t\in\mathbb R$, $\bm\pi_t=\bm\pi_t^{\LMP}+\Delta_t\boldsymbol 1$, and the demand payment becomes $\sum_{m\in[M]}\pi_{mt}^{\LMP}\hat d_t^m+D_t\Delta_t$.
Constraint~\eqref{eq:nodal_zero_uplift} requires
\[ \Delta_t\geq \max_{(i,m)\in\mathcal A_t} \left(c_{it}^m-\pi_{mt}^{\LMP}\right) =\Delta\pi_t^*.
\]
Since $D_t>0$, the unique minimum is attained at $\Delta_t=\Delta\pi_t^*$, yielding \eqref{eq:LMPU}. Property 1) follows directly from \eqref{eq:nodal_zero_uplift} and \eqref{eq:congestion_spread}.

{For Property 2), the LMP expression, energy balance, and complementary slackness of the transmission constraints give $\mbox{\sf MS}_t(\bm\pi_t^{\LMP}) =(\bm\phi_t^*)^\top\bar{\bm f}\geq0$. The common adder changes demand and generator payments by the same amount under energy balance, so LMP+U preserves this surplus. Since LMP+U also eliminates MWPs, demand payments cover generator payments in the modeled energy settlement.}

For Property 3), if \eqref{eq:flow} is nonbinding, then $\pi_{mt}^{\LMP}=\lambda_t^*$ at every bus, and
\begin{align*}
\pi_t^{\LMPUm} =\lambda_t^* +\max_{(i,m)\in\mathcal A_t} \left(c_{it}^m-\lambda_t^*\right) =\max_{(i,m)\in\mathcal A_t}c_{it}^m =\pi_t^{\MDCP}.
\end{align*}
Finally, if the ramping constraints and lower generation bounds are nonbinding, stationarity gives $\pi_{mt}^{\LMP}-c_{it}^m =\overline{\alpha}_{it}^{m*}\geq0$ for every active generator. If one is below its upper capacity, complementary slackness gives $\overline{\alpha}_{it}^{m*}=0$. Therefore, $\Delta\pi_t^*=0$ and $\bm\pi_t^{\LMPUm}=\bm\pi_t^{\LMP}$.
\hfill$\blacksquare$


{Although LMP+U preserves nodal price differences within each interval, the adder can vary over time and therefore change price differences between consecutive intervals.}

\section{Numerical Results}

We evaluate the proposed LMP with Uniform Price Adder  (\LMPU{}) against LMP and two alternative pricing schemes using a rolling-window, multi-interval simulation.  The analysis considers aggregate generator net revenue, {demand payments,} make-whole payments (MWPs), MWP-adjusted merchandising surplus, bus-price volatility, congestion-charge distortion, and bidding incentives.

\subsection{Parameter Settings}

We simulate the three-bus system in Fig.~\ref{fig:three_bus} with three 500-MW generators, $G_1$, $G_2$, and $G_3$, whose marginal costs are \$25/MWh, \$30/MWh, and \$50/MWh, respectively.  Demand is located at Bus~3 and follows the 100-day CAISO net-demand profiles shown in Fig.~\ref{fig:caiso_profile}. The dispatch uses a four-interval rolling window ($W=4$), corresponding to a one-hour look-ahead horizon at 15-minute resolution\footnote{Further details on the CAISO data and rolling-window simulation are provided in~\cite{Chen26arxiv}.}.  We vary the generators' common 15-minute ramping capability as a fraction of their capacity and compare uncongested transmission with a congested case in which $\overline f_{13}=200$ MW.

We compare \LMPU{} with LMP, Extended MDCP (E-MDCP), and PMP. At each generation bus, E-MDCP sets the price to the maximum of its LMP and the highest bid among its dispatched generators; it retains LMP at buses without dispatched generation. PMP uses previously settled prices in a separate multi-interval pricing problem intended to reduce out-of-market uplift \cite{Chen21TPS}. Any remaining bid-cost shortfall under LMP or PMP is compensated through an MWP. Constraint-preserving multi-interval pricing (CMP) was also evaluated but is omitted because it coincided with LMP in these simulations.\footnote{{Supplementary IEEE RTS-24 results are included in the extended version~\cite{PESIM}.}}

\begin{figure}[htbp]
    \centering
    \begin{minipage}[b]{0.49\columnwidth}
        \centering
        \raisebox{0.5cm}{\includegraphics[width=\linewidth]{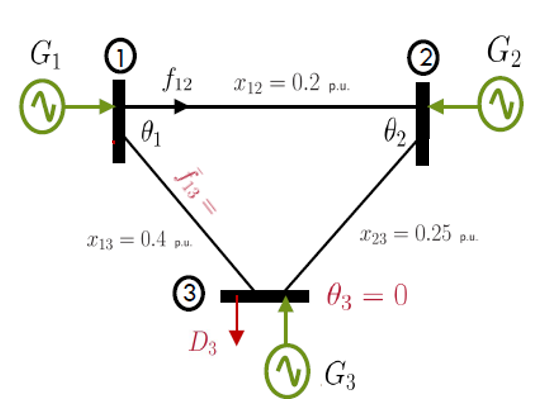}}
        
        \vspace{-1.2em} 
        
        \caption{Topology of the 3-bus network.}
        \label{fig:three_bus}
        
        \vspace{-0.1em} 
    \end{minipage}
    \hfill
    \begin{minipage}[b]{0.49\columnwidth}
        \centering
        \includegraphics[width=\linewidth]{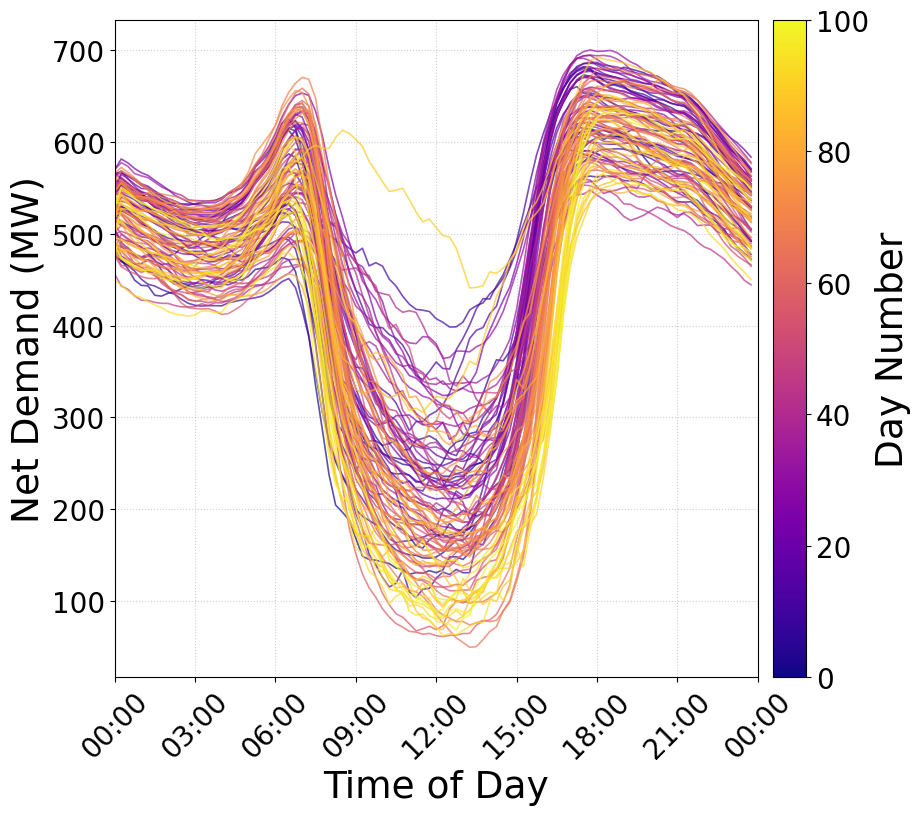}
        
        \vspace{-1.2em} 
        
        \caption{Scaled 100-day CAISO net-demand profile.}
        \label{fig:caiso_profile}
        
        
        \vspace{-0.1em} 
    \end{minipage}
\end{figure}

\subsection{{Generator Revenue, Demand Payments, and Merchandising Surplus}}

Figure~\ref{fig:financial_outcomes} reports 100-day average daily performance metrics over the selected morning ramp-down hours. Aggregate generator net revenue is defined as energy revenue plus MWP minus production cost, while the third column reports MWP-adjusted merchandising surplus, defined as merchandising
surplus after deducting OOM MWPs but before real-time market obligations. 


 \begin{figure}[htbp]
    \centering
    \begin{minipage}{\columnwidth}
        \centering
        \includegraphics[width=\textwidth]{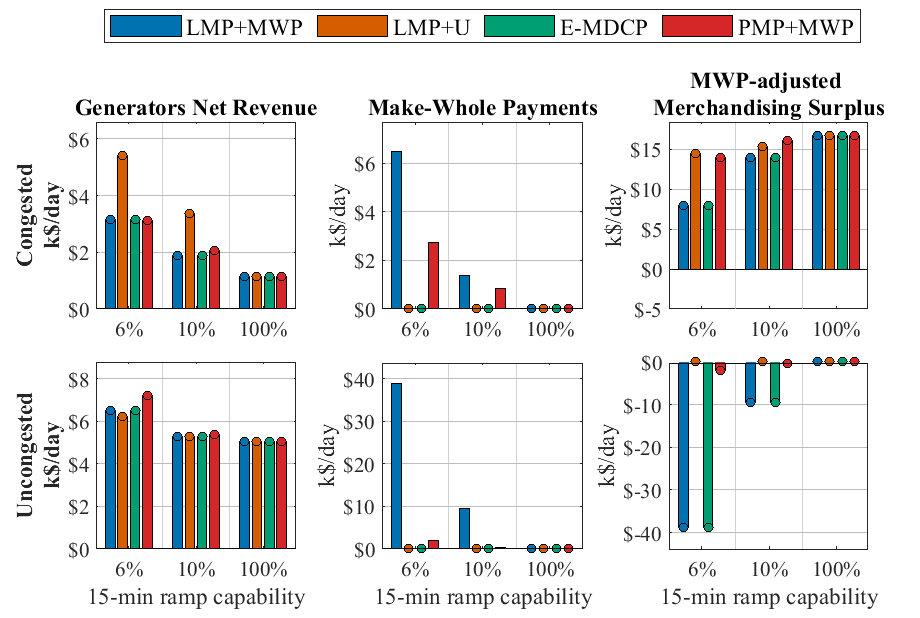}
        
        \vspace{-1.2em} 
        
        \caption{100-day average daily generator net revenue, MWPs, and MWP-adjusted merchandising surplus during the morning ramp-down: congested (top) and unconstrained (bottom).}
        \label{fig:financial_outcomes}
        
    \end{minipage}
\end{figure}

Figure~\ref{fig:financial_outcomes} shows that \LMPU{}  generally produces higher aggregate generator net revenue than LMP with MWP because a positive common price adjustment increases the net revenue of every active generator. However, because the \LMPU{} price adjustment is unrestricted in sign, an exception occurs during the unconstrained 6\% ramping scenario, where a negative adjustment yields slightly lower net revenue.

The second column shows the MWP requirements accumulated over the morning window. Without congestion, production is concentrated among lower-cost generators that lack sufficient downward flexibility, resulting in roughly \$39k/day in MWPs under LMP at 6\% ramping capability. With congestion, production is distributed across more generators. Although this increases total production cost, it provides additional downward flexibility, reducing MWPs under LMP to approximately \$6.5k/day. PMP reduces MWPs but does not eliminate them. As ramping capability increases to 100\%, MWPs decrease and the metrics across mechanisms converge.

Finally, the third column confirms the theoretical properties established in Theorem~\ref{thm1}. Both \LMPU{} and E-MDCP eliminate MWPs entirely, but only \LMPU{} is guaranteed to preserve the nonnegative LMP merchandising surplus. Under LMP and PMP, the MWP-adjusted merchandising surplus becomes negative when the required MWPs exceed the collected surplus, indicating a negative settlement balance even before real-time market obligations are considered.

{Total demand payments include energy payments and MWPs, paid by demand at Bus~3. At 6\% ramping capability, morning payments fall from \$76.30k under LMP+MWP to \$72.40k under \LMPU{} without congestion (5.11\%), but rise from \$90.39k to \$92.61k with congestion (2.46\%). At 10\%, the changes are $-0.28\%$ and $+1.68\%$; payments coincide at 100\%. Eliminating MWPs can therefore offset higher energy payments. With demand at multiple buses, the adder changes energy payments in proportion to each bus's consumption within the interval.}
\vspace{-0.1cm}
\subsection{Price Volatility and Congestion-Charge Distortion}

Figure~\ref{fig:price_volatility_distortion} compares the hourly price volatility and congestion-charge distortion under congested transmission with 6\% ramping capability. {Volatility is measured by the standard deviation of the Bus~3 price across the 100 simulated days at each time of day.}  Congestion-charge distortion is measured by the RMSE between
each mechanism's Bus~3--Bus~1 price spread and the corresponding
LMP spread.

 \begin{figure}[htbp]
    \centering
    \begin{minipage}{\columnwidth}
        \centering
        \includegraphics[width=\textwidth]{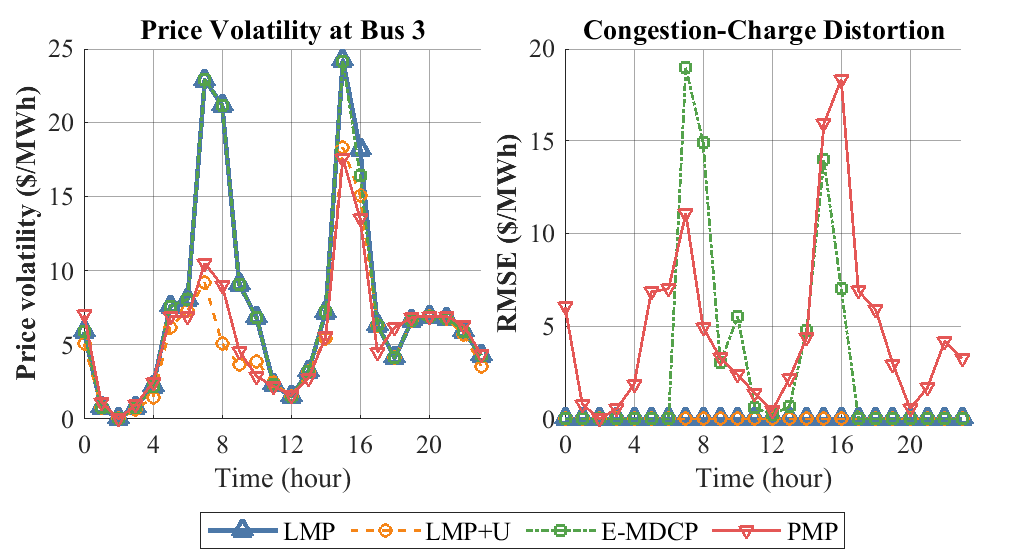}
        
        \vspace{-1.2em} 
        
        \caption{Hourly Bus~3 price volatility (left) and Bus~3--Bus~1 congestion-charge distortion (right) under congested transmission and 6\% ramping capability.}
        \label{fig:price_volatility_distortion}
        
    \end{minipage}
\end{figure}

Price volatility is greatest during the morning and afternoon ramps. {In this case, \LMPU{} reduces the afternoon peak of this standard deviation from approximately \$24/MWh to \$17/MWh, or about 29\%.} E-MDCP and PMP also reduce volatility in some hours.
However, because they adjust prices differently across buses, this reduction comes with changes in nodal price spreads that distort the locational signals associated with transmission congestion. \LMPU{} avoids this tradeoff. Consistent with Theorem~\ref{thm1}, it preserves the LMP nodal price spreads and merchandising surplus, yielding zero congestion-charge distortion. In contrast, E-MDCP and PMP guarantee neither property.

\vspace{-0.1cm}
\subsection{Bidding Incentives}
For the bidding experiment, all generators and demand are placed at Bus~3 to remove spatial price differences. The experiment uses a 15-minute ramping capability equal to 6\% of each generator's capacity. {The submitted offers tested for $G_1$ are \$25, \$27, and \$29/MWh, while its true marginal cost remains \$25/MWh.}  All other bids and physical parameters are  fixed.

{Results cover 2:30~p.m.--10:30~p.m., when $G_1$'s dispatch and prices under LMP and \LMPU{} remain unchanged across the tested offers, therefore $G_1$ is a price-taker.}

\begin{table}[t]
\vspace{-0.2cm}
\centering
\caption{True net revenue of $G_1$ over the selected window.}
\label{tab:g1_net_revenue}
\footnotesize
\setlength{\tabcolsep}{4pt}
\begin{tabular}{c c c c c}
\toprule
\multirow{2}{*}{\begin{tabular}[c]{@{}c@{}}$G_1$ Bid-in\\ cost (\$/MWh)\end{tabular}} & \multicolumn{4}{c}{True $G_1$ Net Revenue (\$)} \\
\cmidrule(lr){2-5}
 & \textbf{LMP+MWP} & \textbf{LMP+U} & \textbf{E-MDCP} & \textbf{PMP+MWP} \\
\midrule
25 & 18,208 & 18,602 & 18,602 & 20,199 \\
27 & 18,366 & 18,602 & 18,602 & 21,004 \\
29 & 18,524 & 18,602 & 18,602 & 21,808 \\
\bottomrule
\end{tabular}
\end{table}

Table~\ref{tab:g1_net_revenue} shows that, although the LMP and dispatch remain unchanged as $G_1$ raises its offer, its true net revenue under LMP+MWP increases from \$18.2k to \$18.5k because the MWP is calculated from its submitted offer. The bid-dependent MWP therefore creates an incentive to inflate the submitted offer. True net revenue under PMP+MWP also increases with the submitted bid. In contrast, true net revenue remains constant at \$18.6k under \LMPU{} and E-MDCP. {This result holds while dispatch and the highest offer among other dispatched generators remain fixed and exceed $G_1$'s offer \cite{Chen26arxiv}. It does not cover deviations that change dispatch or allow $G_1$ to determine the price.}

\vspace{-0.5em}
\section{Conclusion}
{This paper presents \LMPU{}, a uniform pricing rule that eliminates energy make-whole uplift at minimum demand payment while preserving LMP congestion charges in multi-interval dispatch with ramp-constrained generators and inaccurate forecasts. It preserves nonnegative merchandising surplus and ensures revenue adequacy for the modeled real-time energy settlement. Simulations show that OOM make-whole payments can incentivize offer inflation even when dispatch and prices remain fixed; this incentive is absent under \LMPU{} within the tested bid range. The results also show lower price-volatility peaks in the cases examined. Although uniform energy payments can increase, eliminating MWPs can offset this increase for consumers.}

\section*{Appendix: Additional IEEE RTS-24 Results}

\subsection{Simulation Settings}

We extend the numerical comparison to the 12-generator IEEE RTS-24 system using generator data and nodal load shares from \cite{Ordoudis2016RTS24}. Figure~\ref{fig:rts24_network} shows the network, generator costs, and cost groups. The 100 CAISO net-demand trajectories are rescaled to match the magnitude of the system's total generation capacity and distributed across buses using the reference nodal load shares.

Dispatch uses a four-interval rolling window with 15-minute intervals. Upward and downward ramping limits are set to 6\%, 10\%, or 100\% of each generator's capacity per interval. Thus, generators share the same ramping percentage but have different limits in MW. These experiments use the energy-only dispatch formulation of the main paper.

The uncongested case uses the original transmission limits from \cite{Ordoudis2016RTS24}. The congested case reduces the limits on lines 15--21 and 16--17 from 1,000 and 500~MW to 450 and 400~MW, respectively, as marked in Fig.~\ref{fig:rts24_network}.

\begin{figure}[H]
\centering
\begin{minipage}[c]{0.50\columnwidth}
    \centering
    \includegraphics[width=\linewidth]{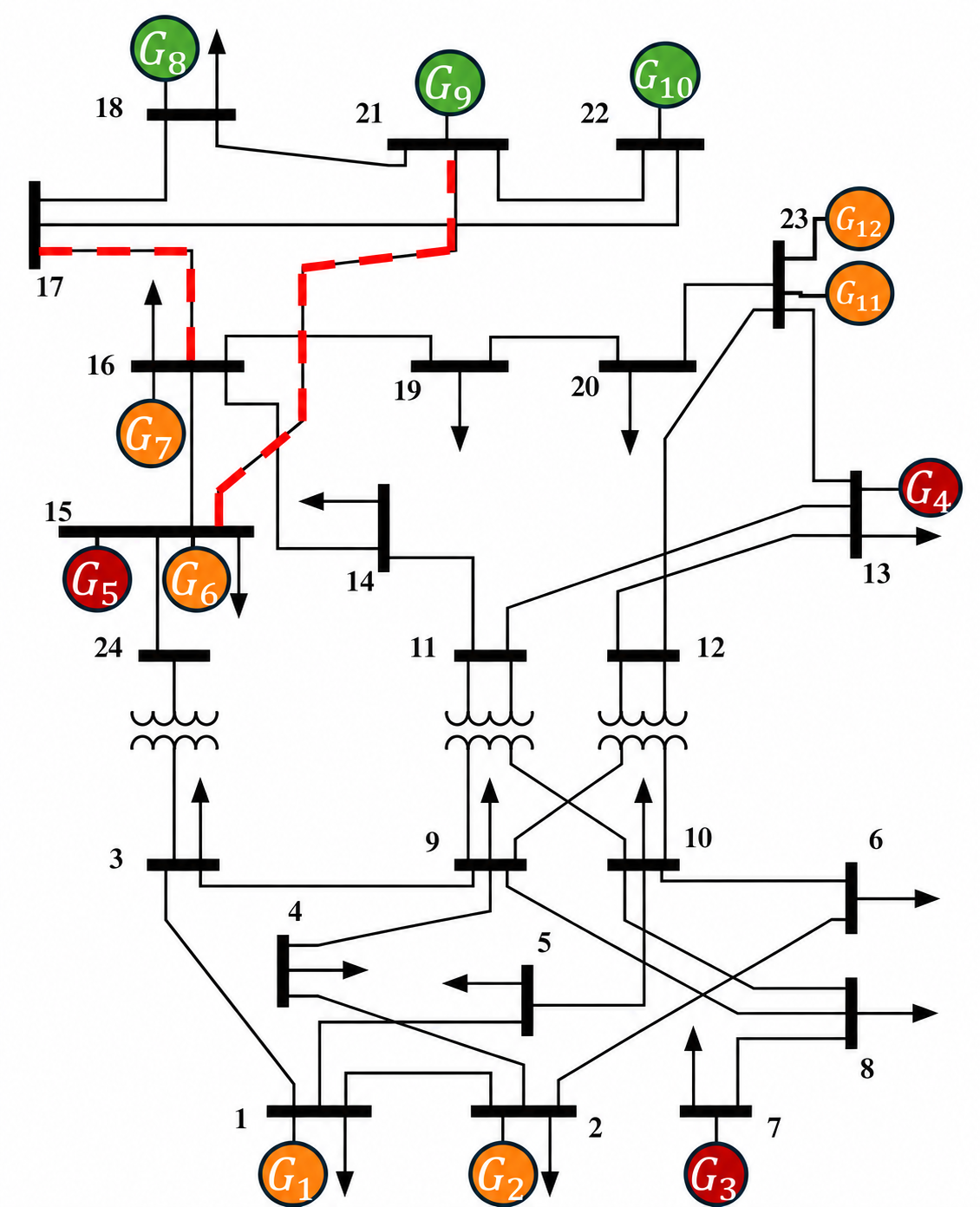}
\end{minipage}\hfill
\begin{minipage}[c]{0.4\columnwidth}
    \centering
    \scriptsize
    \setlength{\tabcolsep}{2pt}
    \renewcommand{\arraystretch}{0.93}
    \begin{tabular}{@{}crc@{}}
        \toprule
        Generator & Cost & Group \\
                  & (\$/MWh) & \\
        \midrule
        $G_1$    & 13.32 & Medium \\
        $G_2$    & 13.32 & Medium \\
        $G_3$    & 20.70 & High \\
        $G_4$    & 20.93 & High \\
        $G_5$    & 26.11 & High \\
        $G_6$    & 10.52 & Medium \\
        $G_7$    & 10.52 & Medium \\
        $G_8$    &  6.02 & Low \\
        $G_9$    &  5.47 & Low \\
        $G_{10}$ &  0.00 & Low \\
        $G_{11}$ & 10.52 & Medium \\
        $G_{12}$ & 10.89 & Medium \\
        \bottomrule
    \end{tabular}
\end{minipage}
        \vspace{-0.3em} 
\caption{IEEE RTS-24 network, marginal production costs, and generator groups. Green, orange, and red indicate low-, medium-, and high-cost generators; dashed red lines indicate the lines with reduced limits in the congested case.}
\label{fig:rts24_network}
\end{figure}

\subsection{Generator Revenue and Demand Payments}

We compare LMP+MWP, \LMPU{}, E-MDCP, and PMP+MWP using the same dispatch for each settlement. Financial results are averaged over 100 complete days. Generator,net revenue includes energy revenue and MWPs less production cost, total demand payments include energy payments and MWPs. Merchandising surplus after deducting MWPs excludes other real-time market obligations.

Figure~\ref{fig:financial_outcomes2} shows that \LMPU{} eliminates MWPs while preserving LMP merchandising surplus. At 6\% ramping capability, aggregate generator net revenue is lower under \LMPU{} than under LMP+MWP in both transmission limit cases, illustrating that the revenue comparison depends on operating conditions and period of the day.

 \begin{figure}[htbp]
    \centering
    \begin{minipage}{\columnwidth}
        \centering
        \includegraphics[width=\textwidth]{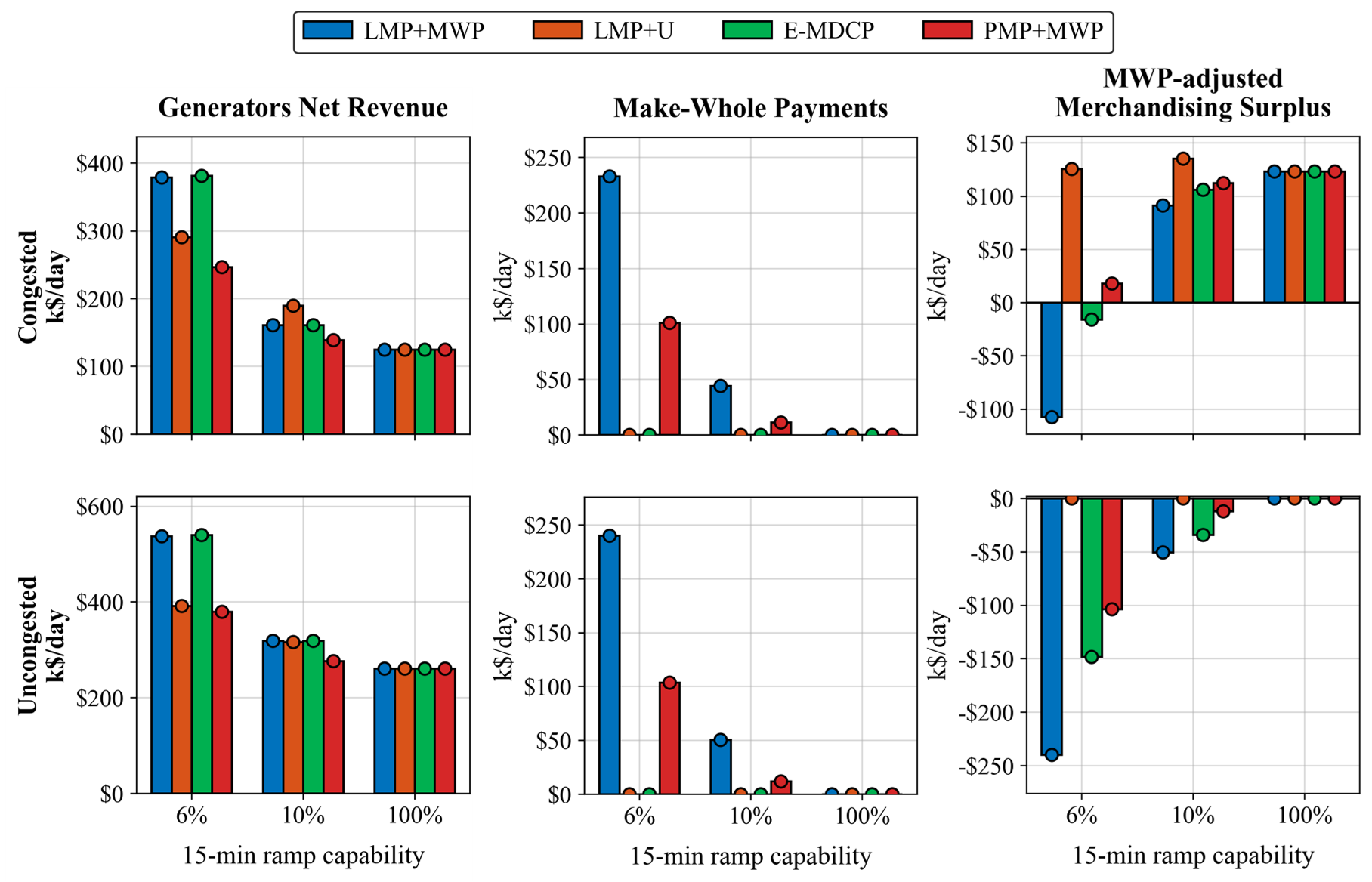}
        
        \vspace{-1.2em} 
        
        \caption{Mean daily generator net revenue, MWPs, and merchandising surplus after deducting MWPs on IEEE RTS-24: congested (top) and uncongested (bottom). Results cover the full day and are averaged over 100 demand trajectories.}
        \label{fig:financial_outcomes2}

    \end{minipage}
\end{figure}

Table~\ref{tab:rts24_demand_payments} reports mean total demand payments over the full day. At 6\% ramping capability, \LMPU{} reduces payments relative to LMP+MWP by 20.61\% without congestion and 14.32\% with congestion. At 10\%, payments decrease by 0.64\% without congestion but increase by 4.54\% with congestion. Thus, savings from eliminating MWPs offset the increase in energy payments in three of these four cases.

Payments under \LMPU{} exceed those under PMP+MWP at both ramping levels, while the comparison with E-MDCP varies across cases. This is consistent with the scope of the optimality result: \LMPU{} minimizes energy payments subject to both energy bid-cost recovery and preservation of nodal LMP differences. All four mechanisms coincide at 100\% ramping capability.

\begin{table}[!t]
\centering
\caption{Mean total demand payments on IEEE RTS-24
(k\$/day), including energy payments and MWPs.}
\label{tab:rts24_demand_payments}
\small
\setlength{\tabcolsep}{3pt}
\begin{tabular}{@{}lrrrr@{}}
\toprule
Ramp & LMP+MWP & LMP+U & E-MDCP & PMP+MWP\\
\midrule
\multicolumn{5}{c}{Uncongested}\\
6\%   & 893.96 & 709.75 & 740.25 & 703.89\\
10\%  & 625.09 & 621.08 & 591.45 & 580.59\\
100\% & 558.48 & 558.48 & 558.48 & 558.48\\
\midrule
\multicolumn{5}{c}{Congested}\\
6\%   & 864.01 & 740.27 & 714.86 & 694.39\\
10\%  & 607.91 & 635.48 & 578.96 & 572.12\\
100\% & 552.55 & 552.55 & 552.55 & 552.55\\
\bottomrule
\end{tabular}
\end{table}

\subsection{Price Volatility and Congestion Charges}

Figure~\ref{fig:price_volatility_distortion2} compares price volatility and congestion-charge distortion under congested transmission and 6\% ramping capability. Volatility is the standard deviation of the Bus~3 price across 100 simulated days at each hour. Under LMP, it reaches approximately \$66/MWh around 07:00 and \$75/MWh around 15:00. At those hours, \LMPU{} lowers the standard deviation to approximately \$6/MWh and \$22/MWh, respectively. The reduction extends beyond these peaks: between roughly 08:00 and 14:00, \LMPU{} remains below LMP. Prices are less variable under all mechanisms overnight and during the late evening, when their volatility profiles are similar.

The comparison also distinguishes the effects of the pricing mechanisms. E-MDCP closely follows LMP's volatility despite producing substantial congestion-charge distortion, with RMSE peaks of approximately \$70/MWh around 08:00 and \$48/MWh around 15:00. PMP reduces volatility during several morning hours but retains an afternoon peak of approximately \$64/MWh and produces small congestion-charge deviations. In contrast, \LMPU{} reduces both volatility peaks while preserving the Bus~3--Bus~1 LMP difference exactly. These results demonstrate that lower price variability and preservation of spatial price differences can be achieved together in this case.

 \begin{figure}[htbp]
    \centering
    \begin{minipage}{\columnwidth}
        \centering
        \includegraphics[width=0.6\textwidth]{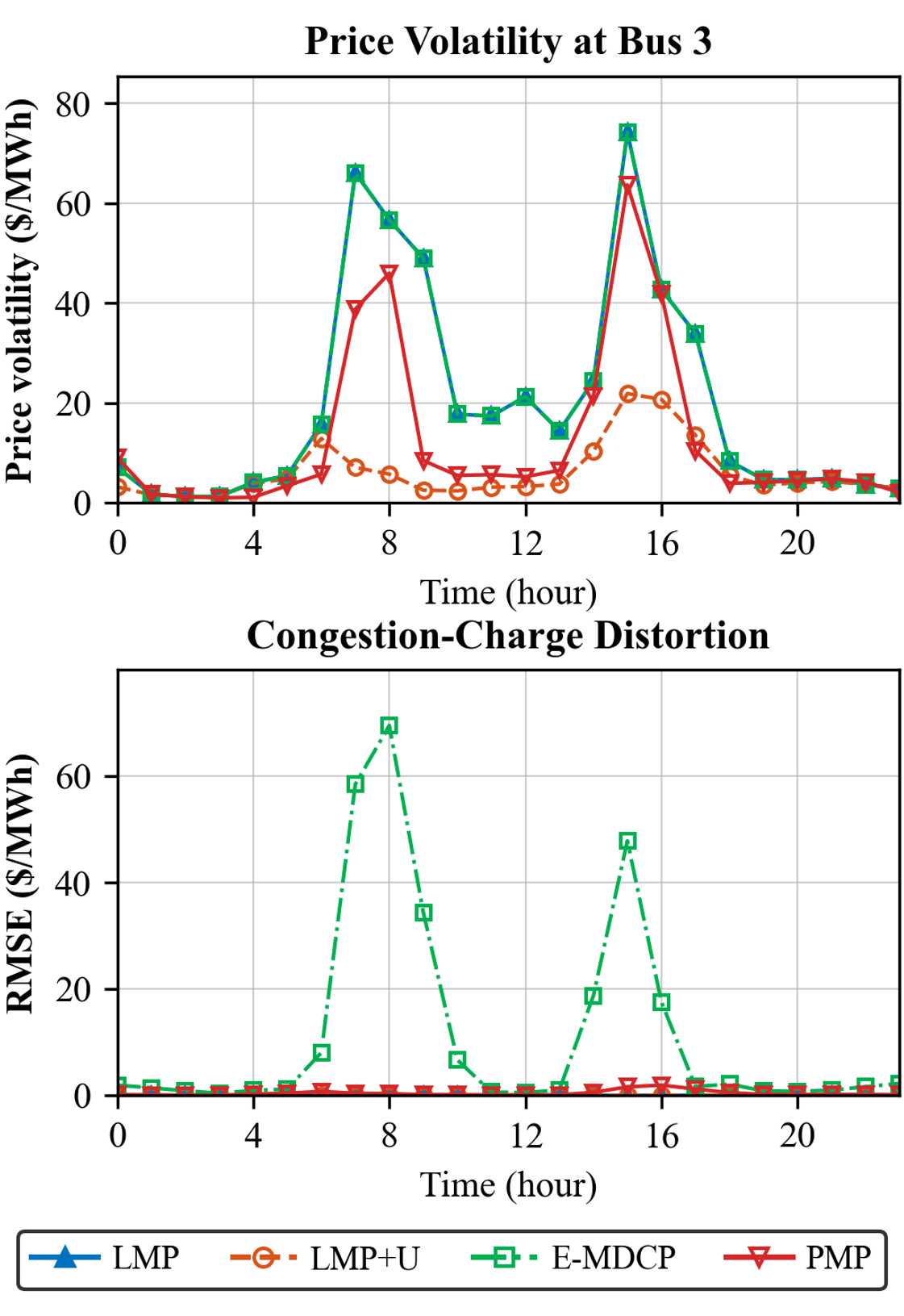}
        
        \vspace{-1.2em} 
        
        \caption{Bus~3 price volatility across 100 demand trajectories (top) and Bus~3--Bus~1 congestion-charge distortion (bottom) on IEEE RTS-24 under congested transmission and 6\% ramping capability.}
        \label{fig:price_volatility_distortion2}
        
    \end{minipage}
\end{figure}

\end{document}